\documentclass[12 pt]{amsart}
\usepackage{graphicx}

\usepackage{amsthm,amsmath,latexsym,amssymb,amsfonts,amscd,hyperref}
\usepackage{color,tikz,tikz-cd}
\usepackage[all]{xy}
\newtheorem{theorem}{Theorem}[section]
\newtheorem{lemma}[theorem]{Lemma}
\newtheorem{proposition}[theorem]{Proposition}

\newtheorem{cor}[theorem]{Corollary}
\theoremstyle{definition}
\newtheorem{definition}[theorem]{Definition}

\theoremstyle{remark}
\newtheorem{remark}[theorem]{Remark}

\numberwithin{equation}{section}

\newcommand{\ww}{{\mathcal{W}}}

\newcommand{\G}{{\mathcal{G}}}

\begin{document}
\phantom{a}
\vspace{-1.5cm}
\title[Deformation quantization of contact manifolds]{Deformation quantization of contact manifolds}

\author{Boris M. Elfimov}
\address{Physics Faculty, Tomsk State University, Lenin ave. 36, Tomsk 634050, Russia}
\email{e1fimov@mail.ru}
\author{ Alexey A. Sharapov }
\address{Physics Faculty, Tomsk State University, Lenin ave. 36, Tomsk 634050, Russia}
\email{sharapov@phys.tsu.ru}
\thanks {The work was supported 
by  the Foundation for the Advancement of Theoretical Physics and Mathematics ``BASIS''}


\subjclass[2010]{Primary 53D55; Secondary 53D35}


\keywords{Contact manifolds, Jacobi brackets, deformation quantization, quantum anomalies, invariants of contact manifolds}

\begin{abstract}
We extend Fedosov deformation quantization to general contact manifolds. Unlike the case of symplectic manifolds, not every classical observable on a contact manifold is generally quantized. On examination of possible obstructions to quantization, we obtain some new invariants of contact manifolds.  
\end{abstract}

\maketitle


\section{Introduction}

In \cite{Fedosov:1994}, B. Fedosov proposed an elegant geometric method for deformation quantization of general symplectic and regular Poisson manifolds. Since then the method has been discussed, generalized, and applied by many authors in various contexts. 
In this paper, we propose yet another extension of Fedosov quantization, this time to the case of contact manifolds.

Contact manifolds form an important class of presymplectic manifolds where the kernel of a presymplectic form is one-dimensional. Given considerable interest in deformation quantization in both the physical and mathematical communities, it seems surprising that not much literature has been devoted to the deformation quantization of general presymplectic manifolds. Here we should mention the paper \cite{Vaisman_2002}, where the problem was addressed in the framework of symplectic ringed spaces. Among other things, the work \cite{Vaisman_2002} revealed some geometric obstructions to the existence of a generalized abelian connection, a key ingredient of Fedosov quantization.
In more explicit terms, these obstructions have also been discussed in our recent paper \cite{GorElfSha:2020}. From this viewpoint the case of contact manifolds is very special: due to the one-dimensionality of the kernel, the aforementioned obstructions are absent. This allows one to perform `half' of the Fedosov quantization procedure. The remaining half, related to the construction of quantum observables, may still run into obstacles. In other words, not every classical observable can be promoted to the quantum level, i.e., quantized. It is the presence of `unquantizable' classical observables that distinguishes  the deformation quantization of contact manifolds form that of symplectic ones.  Finding these obstructions in an explicit form and clarification of their geometric meaning is one of the purposes of the present paper. As a by-product, we obtain some new invariants of contact manifolds, which may be of interest in their own right. 

The problem of deformation quantization of contact manifolds is closely related to the quantum reduction of exact symplectic manifolds by  Hamiltonian constraints. Recall that a symplectic manifold $(N,\omega)$ is called exact if its symplectic form can be written as $\omega=d\lambda$ for some one-form $\lambda$. Under certain regularity conditions, the restriction of the symplectic potential $\lambda$ to a codimension-one submanifold $M\subset N$, defined by a single equation $H=0$, gives a contact structure on $M$.  The standard Hamiltonian reduction by the constraint $H=0$ leads then to the `reduced phase space' $N/\!/H$, which may not be a smooth manifold anymore. In this situation, one can define the algebra of quantum observables on $M$ by means of a suitable  quantum reduction of that on the ambient symplectic manifold $N$. As was shown in \cite[Sec. 7]{Bordemann_2000}, even in a nice case where $N/\!/H$ is a smooth manifold, the reduced quantum algebra may be `much smaller' than the algebra of smooth functions on $N/\!/H$. This is consistent with the general conclusion about possible abstractions to quantization on contact manifolds. 

The rest of the paper is organized as follows. In the next Sec.~2,
we briefly recall the basic definitions and constructions of contact geometry. In particular, we discuss the space of contact connections. Sec.~3 is concerned with two main ingredients of Fedosov deformation quantization: the Weyl algebra bundle and the abelian connection. In Sec.~4, we formulate the deformation quantization problem for contact manifolds and solve it by means of Fedosov's machinery. The final Sec.~5 is devoted to a closer examination of obstructions to quantization. As a by-product of our analysis, we obtain new invariants of contact manifolds: a character  $\chi$ and its periods over closed characteristics.

\section{Contact manifolds and connections}\label{S2}
Let $M$ be a smooth manifold of odd dimension $n=2m+1$. By $\Lambda(M)=\bigoplus \Lambda^p(M)$ and $V(M)=\bigoplus V_p(M)$ we denote, respectively, the exterior algebra of differential forms and the Gerstenhaber algebra of  polyvector fields on $M$. A one-form $\lambda\in \Lambda^1(M)$  is called {\it contact}, if $v=\lambda\wedge (d\lambda)^m$ is a volume form on $M$. Since $v\neq 0$ at each point of $M$,  the differential $\omega=d\lambda$ defines an exact presymplectic 
two-form on $M$ of rank $2m$. The pair $(M,\lambda)$ is called a {\it contact manifold} \cite{Blair}. 

Associated to the contact one-form $\lambda$ is the {\it Reeb vector field } $\xi\in V_1(M)$. This is defined uniquely by the equations 
$$
i_\xi\omega=0\,,\qquad i_\xi\lambda=1\,.
$$

The classical theorem of Darboux  states that every contact manifold has local coordinate charts $\{U; t, q^a, p_a\}_{a=1}^m$ which are canonical, in the sense that 
\begin{equation}\label{DC}
\lambda|_U=dt+p_adq^a\,,\qquad \omega|_U=dp_a\wedge dq^a\,,\qquad \xi|_U=\frac{\partial}{\partial t}\,. 
\end{equation}

The natural volume form  $v$ on a contact manifold defines canonical isomorphisms $\phi_p:\Lambda^p(M)\rightarrow V_{n-p}(M)$ between the spaces of differential forms and polyvector fields  of complementary degrees.  In particular, applying the map $\phi_{n-2}$ to the exterior product $\lambda\wedge \omega^{m-1}\in \Lambda^{n-2}(M)$ yields a bivector $\pi\in V_2(M)$. In terms of local coordinates $\{U; x^i\}_{i=1}^n$, one can write
\begin{equation}\label{defPi}
\pi^{ij}=2m\frac{\varepsilon^{jik_1\cdots k_{n-2}}\lambda_{k_1}\omega_{k_2k_3}\cdots\omega_{k_{n-3}k_{n-2}}}{\varepsilon^{l_1\cdots l_{n}}\lambda_{l_1}\omega_{l_2 l_3}\cdots\omega_{l_{n-1}l_n}}\,,\qquad \xi^i=\frac{\varepsilon^{ik_1\cdots k_{n-1}}\omega_{k_1k_2}\cdots\omega_{k_{n-2}k_{n-1}}}{\varepsilon^{l_1\cdots l_{n}}\lambda_{l_1}\omega_{l_2 l_3}\cdots\omega_{l_{n-1}l_n}}\,,
\end{equation}
where $\varepsilon^{i_1\cdots i_n}$ is the Levi-Civita symbol. It follows from the definition that $\lambda_i\pi^{ij}=0$ and the matrix $\pi^{ij}(p)$ has rank $2m$ at each point $p\in M$.  Let us introduce the tensor $P^i_j=\delta_j^i-\lambda_j\xi^i$, which defines a projector of rank $2m$
on the tangent bundle of $M$, i.e., $P^2=P$ and $\mathrm{Tr}P=2m$. 
Then the following equalities hold:
$$
P_j^i=\pi^{ik}\omega_{kj}\,,\qquad \lambda_iP^i_j=0\,,\qquad P_j^i\xi^j=0\,,\qquad P_i^k\omega_{kj}=\omega_{ij}\,,\qquad \pi^{ik}P_k^j=\pi^{ij}\,.
$$
In Darboux coordinate system (\ref{DC}), the bivector $\pi$ takes the form
$$
\pi|_U=\Big(\frac{\partial}{\partial q^a}-p_a\frac{\partial}{\partial t}\Big)\wedge \frac{\partial}{\partial p_a}\,.
$$
Computing now the Schouten--Nijenhuis brackets of $\pi$ and $\xi$, one can easily find 
\begin{equation}\label{Jac}
[\pi,\pi]=2\pi\wedge\xi\,,\qquad [\xi,\pi]=0\,.
\end{equation}
These relations identify the triple $(M,\pi,\xi)$ as a {\it Jacobi manifold} \cite{kirillov31local}, \cite{lichnerowicz1978varietes}. Thus, every contact manifold carries the canonical structure of a Jacobi manifold. Using the Jacobi structure, one can endow the space of smooth functions $C^\infty(M)$ with the following Lie bracket, called  also a {\it Jacobi bracket}:
\begin{equation}\label{bracket}
\{f,g\}=\pi(df,dg)+f\xi g -g\xi f\,.
\end{equation}
Relations (\ref{Jac}) ensure the Jacobi identity for the Lie bracket. The bracket (\ref{bracket}) is not a Poisson bracket as it fails to satisfy the Leibniz identity:
$$
\{f,gh\}-\{f,g\}h-g\{f,h\}=gh\xi f\,.
$$
Notice, however, that the Jacobi bracket gives rise to a Poisson algebra structure on the commutative subalgebra of $\xi$-invariant functions $C^\infty_\xi(M)\subset C^\infty(M)$.  

{By a deformation quantization of a contact manifold $(M,\lambda)$ we mean a deformation quantization of the Poisson algebra $\big(C^\infty_\xi(M), \{\,\cdot\,,\,\cdot\,\}\big)$.} As usual, this requires a suitable connection on the underlying manifold $M$. 

By definition \cite{LICHNEROWICZ1982}, a {\it contact connection} on $(M, \lambda)$ is a symmetric affine connection $\nabla$ such that 
$$
\nabla_X \omega=0\,,\qquad \nabla_X \xi=0\,\qquad \forall X\in V_1(M)\,.
$$
\begin{theorem}
On every contact manifold there exists a contact connection.
\end{theorem}
\begin{proof}
Take a partition of unity $\{\varphi_\alpha\}_{\alpha\in A}$ subordinated to an atlas $\{U_\alpha\}_{\alpha\in A}$ of Darboux coordinate charts (\ref{DC}). In each coordinate chart $U_\alpha\subset M$, define a symmetric connection $\nabla^{(\alpha)}$ to have zero Christoffel symbols. Clearly, $$\nabla^{(\alpha)}\omega|_{U_\alpha}=0\,,\qquad \nabla^{(\alpha)}\xi|_{U_\alpha}=0\,.$$
The desired contact connection is now obtained by gluing the local connections $\nabla^{(\alpha)}$ with the help of the partition of unity:
$$
\nabla=\sum_{\alpha\in A} \varphi_\alpha\nabla^{(\alpha)}\,.
$$
\end{proof}
We say that a covariant tensor $T_{i_1\ldots i_p}$ of rank $p$ is {\it $\xi$-transverse}, if 
\begin{equation}\label{trans}
\xi^{i_k}T_{i_1\cdots i_k\cdots i_p}=0\qquad \forall k=1,\ldots,p\,.
\end{equation}
Denote by $\mathcal{T}^\xi(M)=\bigoplus\mathcal{T}^\xi_p(M)$ the tensor algebra of covariant $\xi$-transverse tensor fields on $M$. In particular, 
$\omega\in \mathcal{T}^\xi_2(M)$. Each space $\mathcal{T}^\xi_p(M)$ contains the subspace $\mathcal{S}^\xi_p(M)$ 
of fully symmetric covariant tensor fields of rank $p$. Combining the tensor product with the complete symmetrization of indices makes the space $\mathcal{S}^\xi(M)=\bigoplus_p \mathcal{S}^\xi_p(M)$ into a commutative algebra over $C^\infty(M)$. By definition, the algebra $\mathcal{S}^\xi(M)$ is naturally graded by the rank of symmetric tensors.

With the definitions above, we are ready to describe the space of all contact connections. 
\begin{theorem}
The contact connections form an affine space $\mathrm{Con}(M,\lambda)$ modelled on the linear space $\mathcal{S}^\xi_2(M)\oplus \mathcal{S}^\xi_3(M)$.

\end{theorem}
\begin{proof}
As usual the difference of two affine symmetric connections $S=\nabla'-\nabla$ is given by a tensor $S_{ij}^k$ which is symmetric in $i$ and $j$. Define the pair of covariant tensors 
\begin{equation}\label{Sdef}
S_{ijk}=S_{ij}^l\omega_{lk}\,,\qquad S_{ij}=S_{ij}^k\lambda_k\,.
\end{equation}
Since $\nabla$ and $\nabla'$ respect $\omega$ and $\xi$, both the tensors are fully symmetric and $\xi$-transverse.
For a given $\nabla$, the assignment $S_{ij}^k\mapsto (S_{ij}, S_{ijk})$ defines a map 
\begin{equation}\label{h}
h_{\nabla}: \mathrm{Con}(M,\lambda)\rightarrow \mathcal{S}^\xi_2(M)\oplus\mathcal{S}^\xi_3(M)
\end{equation}
from the affine space of contact connections to the space of covariant tensor fields. One can easily check that the formula
\begin{equation}\label{S}
S_{ij}^k=S_{ij}\xi^k+S_{ijl}\pi^{lk}
\end{equation}
defines the inverse map. Hence, the map (\ref{h}) is bijective. 
\end{proof}
It is clear that each contact connection induces a connection in the space of covariant $\xi$-transverse tensor fields. In other words, the space  $\mathcal{T}^\xi_p(M)$ is invariant under the action of the covariant derivative $\nabla_X$ for all $X\in V_1(M)$. One can also see that two contact connections $\nabla$ and $\nabla'$ induce the same connection on $\mathcal{T}^\xi(M)$ if and only if $h_\nabla(\nabla'-\nabla)\in \mathcal{S}^\xi_2(M)$. We are lead to conclude  that the space of connections on $\xi$-transverse covariant tensor fields is isomorphic to the space $\mathcal{S}^\xi_3(M)$. 

The proof of the following  useful relations is left to the reader:
\begin{equation}\label{id}
\nabla_Xv=0\,,\quad \omega_{ik}\nabla_X\pi^{kj}=-\xi^j\nabla_X\lambda_i\,,\quad \nabla_X\pi^{ij}=(\pi^{ki}\xi^j-\pi^{kj}\xi^i)\nabla_X\lambda_k
\end{equation}
for all $X\in V_1(M)$.

\section{Weyl algebra bundle and Fedosov connection}
Given a contact manifold $(M,\lambda)$, consider the algebra $\mathcal{S}^\xi(M)\otimes \Lambda(M)$. This is given by the tensor product over  $C^\infty(M)$ of the exterior algebra of differential forms $\Lambda(M)$ and the algebra $\mathcal{S}^\xi(M)$ of fully symmetric, covariant tensor fields on $M$ satisfying the $\xi$-transversality condition (\ref{trans}). For better visualisation of algebra's elements it is convenient to introduce the set of formal commuting  variables $y^i$, $i=1,\ldots, n$, regarded as a local frame of sections in $T^\ast M$.  Then,  the algebra $ \mathcal{S}^\xi(M)\otimes \Lambda(M)$ is spanned by polynomials 
\begin{equation}\label{audx}
a_{j_1\cdots j_qi_1\cdots i_p}(x) y^{j_1}\cdots y^{j_q}dx^{i_1}\wedge \cdots\wedge dx^{i_p}
\end{equation}
in the $y$'s and $dx$'s, where the expansion coefficients $a_{j_1\cdots j_qi_1\cdots i_p}$ constitute a covariant tensor field  on $M$. The tensor is symmetric w.r.t. indices  $j_1,\ldots, j_q$, antisymmetric w.r.t.  $i_1,\ldots, i_p$, and satisfies the $\xi$-transversality condition $\xi^{j_1}a_{j_1\cdots j_qi_1\cdots i_{p}}=0$.  The number of symmetric and antisymmetric indices defines a natural bigrading in the  algebra $ \mathcal{S}^\xi(M)\otimes\Lambda(M) $, so that (\ref{audx}) is a homogeneous element of bidegree $(q,p)$. We will refer to the numbers $q$ and $p$ as the {\it $y$-degree} and {\it form degree}, respectively.

Using the contact structure on $M$, one can construct a one-parameter deformation of the graded-commutative algebra $ \mathcal{S}^\xi(M)\otimes \Lambda(M)$. This is defined by the Weyl--Moyal $\circ$-product:
\begin{equation}\label{c-prod}
    \begin{array}{rcl}
       a \circ b &=&  \displaystyle \left.  \exp \left(\frac{\nu}{2} \pi^{ij} \frac{\partial}{\partial y^i} \frac{\partial}{\partial z^j}\right)  a(x, y, dx) \wedge b(x, z, dx) \right|_{z=y} \\[5mm]
        &= &\displaystyle \sum_{n=0}^\infty \left(\frac{\nu}{2} \right)^n \frac{1}{n!} \pi^{i_1 j_1} \dots \pi^{i_n j_n} \frac{\partial^n a}{\partial y^{i_1} \dots \partial y^{i_n}} \wedge \frac{\partial^n b}{\partial y^{j_1} \dots \partial y^{j_n}} \,.
    \end{array}
\end{equation}
Here  $\nu \in \mathbb{R}$ is the deformation parameter and the tensor $\pi^{ij}$ is given by Eq. (\ref{defPi}). It is easy to see that the  $\circ$-product is associative and does not depend  on the choice of local coordinates. Since $a$ and $b$ are polynomial in $y$'s, the infinite sum above contains only a finite number of nonzero terms. Notice that the $\circ$-product does not respect the first grading unless one treats  $\nu$ as a  formal variable of $y$-degree $2$. Proceeding with this interpretation, we define the algebra $W=\mathbb{R}[\nu]\otimes \mathcal{S}^\xi(M)\otimes \Lambda(M)$ with the $\circ$-product (\ref{c-prod}) extended to $W$ by $\mathbb{R}[\nu]$-linearity. This yields a unital bigraded algebra $W=\bigoplus W_{q,p}$. The unit in  $W$ is the function on $M$ identically equal to $1$.  Let us also introduce the commutator of two elements of $W$ by 
$$
 [a,b]=a\circ b-(-1)^{p_1p_2}b\circ a
$$
for all  $a\in W_{\bullet, p_1}$ and  $b\in W_{\bullet, p_2}$. The commutator makes the $\mathbb{R}[\nu]$-vector space $W$
into a graded Lie algebra w.r.t. the form degree. It is easy to see that the centre of this algebra is generated over  $\mathbb{R}[\nu]$ by  the differential forms of  $\Lambda(M)$.

Associated to the $y$-degree is the descending filtration 
$$
{F}^kW\supset {F}^{k+1}W\,,\qquad {F}^kW=\bigoplus_{q\geq k}W_{q,\bullet}\,,\qquad k=0,1,\ldots
$$
The inverse limit with respect to this filtration,
$$
\ww=\lim_{\leftarrow}(W/{F}^k W)\,,
$$
gives the completion of the algebra $W$. By definition, the elements of the algebra $\ww$ are formal power series in the $y$'s and $\nu$ with coefficients in  $\Lambda(M)$. The algebra $\ww=\bigoplus \ww_{q,p}$ inherits the bigrading, filtration,  and unit of $W$. The centre of $\ww$ is given by $\Lambda(M)[\![\nu]\!]$. 

\begin{remark}
Geometrically, one can think of $\ww$ as the space of sections of an infinite-dimensional vector bundle $\mathcal{E}\rightarrow M$ associated with the tangent bundle $TM$. Formula (\ref{c-prod}) endows $\mathcal{E}$ with a fiberwise associative product, which then induces the associative $\circ$-product in the space of sections $\ww=\Gamma(\mathcal{E})$.  Since at each point $p\in M$ the associative product in $\mathcal{E}_p$ is essentially the Weyl--Moyal product associated with the constant Poisson bivector $\pi^{ij}(p)$, we will refer to  $\mathcal{E}$ as the {\it Weyl algebra bundle}, a term coined in \cite{Fedosov:1994}. 
\end{remark}

Following Fedosov, we also introduce the differential $\delta: \ww_{q,p}\rightarrow \ww_{q-1,p+1}$ of bidegree $(-1,1)$ defined by the formula
$$
\delta a=dx^i\wedge \frac{\partial a}{\partial y^i}=-\frac1\nu[\omega_{ij}y^idx^j, a]\,.
$$
Clearly,
$$
\delta^2=0\quad \mbox{and}\quad \delta (a\circ b)=(\delta a)\circ b+(-1)^{p}a\circ\delta b
$$
for all $a\in \ww_{\bullet, p}$ and $b\in \ww$. In order to describe the cohomology of the operator $\delta$ it is convenient to introduce the subalgebra $\ww^\xi\subset \ww$ of $\xi$-transverse differential forms  and the ideal $\ww^\lambda\subset \ww$ generated by the contact form $\lambda$. These are defined as follows:
\begin{equation}\label{WW}
\ww^\xi=\big\{\;a\in \ww\; \big|\; i_\xi a=0\;\big\}\,,\qquad \ww^\lambda=\big\{ \;\lambda a \;\big|\; \forall a\in \ww \;\big\}\,.
\end{equation}
(Notice that $\lambda a=\lambda \circ a$.) Then $\ww=\ww^\xi\oplus \ww^\lambda$ and each element $a\in \ww$ decomposes uniquely into the sum 
\begin{equation}\label{aa}
a=a^\xi+a^\lambda \,, \qquad a^\xi\in \ww^\xi\,, \quad a^\lambda\in \ww^\lambda\,.
\end{equation}
It is clear that $a^\lambda=\lambda i_\xi a$. Either subspace (\ref{WW}) is obviously invariant under the action of the differential $\delta$. As a result, each group of $\delta$-cohomology $$H_{q,p}(\ww)=\ker (\delta: \ww_{q,p}\rightarrow \ww_{q-1,p+1})/\delta \ww_{q+1,p-1}$$ splits into the direct sum $H_{q,p}(\ww)=H_{q,p}(\ww^\xi)\oplus H_{q,p}(\ww^\lambda)$ of subgroups associated with the subcomplexes $\ww^\xi$ and $\ww^\lambda$.   Define the operator  $\delta^\ast: \ww_{q,p}\rightarrow \ww_{q+1,p-1}$ by
    $$
    \delta^\ast=y^iP_i^j i_{\frac{\partial}{\partial{x^j}}}=y^i(\delta_i^j-\lambda_i\xi^j)i_{\frac{\partial}{\partial{x^j}}}\,.
    $$
It is easy to see  that $( \delta^\ast)^2=0$, $\delta^\ast \lambda=0$, and
    \begin{equation}\label{N}
    (\delta \delta^\ast +\delta^\ast \delta)a_{qp}= (q+p)a_{qp}-\lambda i_\xi a_{qp}
    \end{equation}
for any homogeneous element  $a_{qp}\in \ww_{q,p}$. Unlike $\delta$, the operator $\delta^\ast$ is not a derivation of the $\circ$-product. Define now the operator 
$\delta^{-1}: \ww_{q,p}\rightarrow \ww_{q+1,p-1}$ as  
$$
\delta^{-1}a^\xi_{qp}=\left\{\begin{array}{cc}
   \frac1{q+p}\delta^\ast a_{qp}^\xi \,, &  q+p>0\,;\\[2mm]
     0\,,& q=p=0\,;
\end{array}   \right.
\qquad 
\delta^{-1}a^\lambda_{qp}=\left\{\begin{array}{cc}
   \frac1{q+p-1}\delta^\ast a_{qp}^\lambda \,, &  q+p>1\,;\\[2mm]
     0\,,& q+p=1\,;
\end{array}   \right.
$$
here $a_{qp}=a_{qp}^\xi+a_{qp}^\lambda$ according to (\ref{aa}). Direct verification shows that 
\begin{equation}\label{HD}
a=(\delta \delta^{-1}+\delta^{-1}\delta)a+Pa\,,
\end{equation}
where the operator $P$ acts by the rule
\begin{equation}\label{P}
Pa=a(\nu,x,0,0)+\lambda (i_\xi a)(\nu,x, 0,0) \qquad \forall a(\nu, x, y, dx)\in \ww\,.
\end{equation}
Clearly, $P^2=P$ and $Pi_\xi=i_\xi P$. Formula (\ref{HD}), being similar to the classical Hodge--de Rham decomposition,  means
that the non-trivial $\delta$-cocycles are nested in the kernel of the projector $1-P$.  The space $\ker (1-P)$ is spanned by sums $a+\lambda b$ where $a, b \in C^\infty(M)[\![\nu]\!]$. On the other hand, the elements of the form $a+\lambda b$ are obviously non-trivial $\delta$-cocycles. Thus, we arrive at the following conclusion about the cohomology of the dg-algebra $(\ww, \delta)$. 

\begin{proposition}\label{P31}
$H_{0,0}(\ww)\simeq C^\infty(M)[\![\nu]\!]\simeq  H_{0,1}(\ww)$ and $H_{q,p}(\ww)=0 $ for all $q>0$ or  $p>1$. 
\end{proposition}

Another class of derivations on $\ww$ originates from contact connections on $(M,\lambda)$. Specifically, to each contact connection $\nabla$ on $TM$ there corresponds a connection on the associated bundle of Weyl algebras $\mathcal{E}\rightarrow M$, which we denote by the same symbol $\nabla$. In every coordinate chart, it is defined as
$$
\nabla a =dx^i\wedge \left(\frac{\partial a}{\partial x^i}+y^j\Gamma_{ij}^k\frac{\partial a}{\partial  y^k}\right)\qquad \forall a\in \ww\,,
$$
$\Gamma_{ij}^k(x)$ being the Christoffel symbols. On account of the last equality in (\ref{id}), the following graded Leibniz rule holds:
$$
\nabla (a\circ b)=(\nabla a)\circ b+(-1)^{p}a\circ \nabla b
$$
for all $a\in\ww_{\bullet,p}$ and $b\in\ww$. Since the contact connection is symmetric, the 
operators $\nabla$ and $\delta$ anticommute to each other, i.e., $\nabla\delta+\delta\nabla=0$. Unlike $\delta$, the operator $\nabla$
squares to a nonzero derivation of the form 
$$
\nabla^2 a =\frac1\nu[R, a]\qquad \forall \, a\in \ww\,,
$$
where 
$$
R=\frac14 R_{ijkl}y^iy^jdx^k\wedge dx^l\,, \qquad  R_{ijkl} = \omega_{im} R^m_{jkl}\,,
$$
and $R_{jkl}^m$ is the  curvature tensor of  the contact connection. The Bianchi identity implies $\nabla R=0$.

As in the original Fedosov's construction, we will consider more general connections on the Weyl algebra bundle $\mathcal{E}$, namely, the connections of the form
\begin{equation}\label{D}
Da=\nabla a+ \frac1\nu[\gamma, a]\,,
\end{equation}
where $\gamma\in \ww_{\bullet,1}$. Note that the above formula specifies the one-form $\gamma$ only up to adding a central element  $\theta\in \Lambda^1(M)[\![\nu]\!]$. To avoid this ambiguity one can impose the {\it Weyl normalizing condition} $\gamma|_{y=0}=0$. Even with this restriction on $\gamma$, expression (\ref{D}) does not define the pair $(\nabla, \gamma)$ uniquely. If $\nabla'$ is another contact connection, then the difference $S=\nabla' -\nabla$ is a tensor of the form (\ref{S}). When applying these covariant derivatives to the elements of $\mathcal{S}^\xi(M)$, one may ignore the first term in (\ref{S}) and set  $S^k_{ij}=S_{ijl}\pi^{lk}$ for some tensor $S_{ijk}\in\mathcal{S}_3^\xi(M)$.  Let us define
\begin{equation}\label{gg}
\gamma'=\gamma+ \Delta\gamma \,, \qquad \Delta\gamma=\frac12 S_{ijk}y^iy^j dx^k\,.
\end{equation}
It is easy to see that the pair $(\nabla', \gamma')$ gives the same connection as $(\nabla,\gamma)$. Therefore, without loss in generality, we may assume the contact connection 
$\nabla$ to be chosen once and for all.    Then the covariant derivatives (\ref{D}) are in one-to-one correspondence with one-forms  $\gamma\in \ww_{\bullet,1}$ subject to the Weyl normalizing condition.

Squaring the operator $D$, one can  find that
\begin{equation}\label{WC}
D^2a=\frac1\nu[\Omega, a]\,,\qquad \Omega=R+\nabla \gamma+\frac 1\nu \gamma\circ \gamma
\end{equation}
for all $a\in \ww$. 
The two-form $\Omega$ is called the {\it Weyl curvature} of the connection $D$.  Again, it follows from the Bianchi identity that $D\Omega=0$. The connection $D$ is called {\it abelian} if 
$$
D^2a =\frac 1\nu[\Omega, a]=0 \qquad \forall a\in \ww\,.
$$
This implies that the curvature $\Omega$ of an abelian connection is a closed two-form of $\Lambda^2(M)[\![\nu]\!]$.  
Notice that while the shift of $\gamma$ by a central one-form $\theta$ does not affect the connection (\ref{D}), it changes the Weyl curvature by an exact two-form: $\Omega \rightarrow \Omega+d\theta$.  Therefore only the de Rham class $[\Omega]\in H^2(M)[\![\nu]\!]$ of the Weyl curvature of an abelian connection has an invariant meaning.   

As with symplectic manifolds, the most interesting to us are abelian connections of the form  
\begin{equation}\label{DF}
    D=-\delta +\nabla +\frac1\nu[r, \;\cdot\; ]=\nabla +\frac1\nu[\omega_{ij}y^idx^j+r, \;\cdot\;]\,,
    \end{equation}
where  $r\in F^2\ww$ and $r|_{y=0}=0$. In the following, we will refer to (\ref{DF}) as a {\it Fedosov connection}.
\begin{theorem} \label{T32}
For a given contact connection $\nabla$ and a series $\Omega=\omega+\sum_{k=1}^\infty\nu^k\Omega_{k}$, where $\Omega_{k}$ are closed two-forms on $M$,  
    there is a unique Fedosov connection
    with curvature  $\Omega$ such that  $\delta^{-1}r=0$.
\end{theorem}

\begin{proof}   The Weyl normalizing condition implies that $r|_{y=0}=0$.  Hence, $Pr=0$.  
Eq. (\ref{WC}) for the Weyl curvature of the connection (\ref{DF}) takes the form
$$
\Omega=\omega+R-\delta r+\nabla r+\frac1\nu r\circ r\,.
$$
Instead of solving this equation directly, we apply the operator $\delta^{-1}$ and use decomposition (\ref{HD}) for $r$. This gives
\begin{equation}\label{r}
r=\delta^{-1}(\omega-\Omega+R)+\delta^{-1}\big(\nabla r+\frac1\nu r\circ r\big)\,.
\end{equation}
Since $\delta^{-1}(\omega-\Omega+R)\in F^3\ww$ and $\delta^{-1}$ increases the $y$-degree by one unit, we can solve this equation by iterations to get a 
unique solution for $r\in F^3\ww$. Clearly, $\delta^{-1}r=0$ and $r|_{y=0}=0$. It remains to show that the one-form $r$ obtained in such a way does define a desired abelian connection. Let $D$ be the connection (\ref{DF}) associated with $r$ and let $\bar \Omega$ denote its Weyl curvature 
$$\bar \Omega=\omega+R-\delta r+\nabla r+\frac1\nu r\circ r\,.$$
We apply the operator $\delta^{-1}$ to this equation to get
$$
\delta^{-1}(\bar\Omega-\Omega)=\delta^{-1}(\omega-\Omega+R)+\delta^{-1}\Big(\nabla r+\frac1\nu r\circ r\Big)-\delta^{-1}\delta r=r-\delta^{-1}\delta r=\delta\delta^{-1}r=0\,.
$$
Next, the Bianchi identity $D\bar\Omega=0$ implies $D(\bar\Omega-\Omega)=0$ because the form $\Omega$ is closed and central. 
We can write this as 
$$
\delta(\bar \Omega-\Omega)=(D+\delta)(\bar\Omega-\Omega)\,.
$$
Applying $\delta^{-1}$ yields
$$
\bar \Omega-\Omega=\delta^{-1}(D+\delta)(\bar\Omega-\Omega)
$$
because $\delta^{-1}(\bar\Omega-\Omega)=0$ and $P(\bar\Omega-\Omega)=0$. Since the operator $\delta^{-1}(D+\delta)$ raises the $y$-degree, the iteration method shows that this equation has a unique solution $\bar\Omega-\Omega=0$, cf. Lemma \ref{L43} below. Hence, the one-form $r$, defined by Eq.(\ref{r}), gives  an abelian  connection with curvature $\Omega$. 
\end{proof}
Solving Eq. (\ref{r})  by iterations, one can readily find 
\begin{equation}\label{rR}
r=\Big(\frac16R_{ijkl}y^iy^jy^k\xi^l-2\nu (\Omega_1)_{kl}y^k\xi^l\Big)\lambda+\frac18 R_{ijkl}P^k_mP^l_ny^iy^jy^mdx^n- \nu(\Omega_1)_{ij}P^i_kP^j_ly^kdx^l+\ldots\,,
\end{equation}
where dots stand for the terms of $y$-degree higher than three.  Hence, $r\in F^3\ww$.

Associated with an abelian connection $D$ is the {\it algebra of  flat sections} $\ww_D$. By definition, 
$$
\ww_D=\{a\in \ww_{\bullet,0}\;|\;Da=0\}\,.
$$
Since $D$ differentiates the $\circ$-product, the flat sections form a subalgebra in $\ww$.  The algebra $\ww_D$ is the main object of Fedosov quantization.

The group of all automorphisms of the associative algebra $\mathcal{W}$ contains an important subgroup $\G\subset \mathrm{Aut}(\mathcal{W})$, which is defined as follows. One first introduces an extension $\mathcal{W}^+$ of the algebra $\mathcal{W}$ by formal Laurent series in $\nu$. More precisely, the algebra $\mathcal{W}^+$ consists of formal power series 
\begin{equation}\label{LS}
\sum_{2k+q\geq 0}\nu^k a_{k,j_1\cdots j_qi_1\cdots i_p}(x) y^{j_1}\cdots y^{j_q}dx^{i_1}\wedge \cdots\wedge dx^{i_p}
\end{equation}
where only a finite number of coefficients $a_{k,j_1\cdots j_qi_1\cdots i_p}$ for the negative $k$'s are different from zero. The series (\ref{LS}) form a well-defined algebra w.r.t. the $\circ$-product above and the action of each abelian connection $D$ extends naturally from $\mathcal{W}$ to $\mathcal{W}^+$. The aforementioned group $\G$ consists of invertible elements of $\mathcal{W}^+$ of the form
\begin{equation}
    U=\exp_\circ\Big(\nu^{-1}H\Big)=\sum_{n=0}^\infty \frac{\nu^{-n}}{n!}\underbrace{H\circ H\circ\cdots \circ H}_{n}
\end{equation}
for $H\in F^3\mathcal{W}$. 
Applying the Campbell--Hausdorff formula shows that such elements indeed form a group. It is also clear that the map
\begin{equation}
    a\quad \mapsto\quad  U\circ a\circ U^{-1}=\sum_{n=0}^\infty \frac{\nu^{-n}}{n!}[H,[H,\cdots,[H,a],\cdots]]
\end{equation}
is an automorphism of the algebra $\mathcal{W}$. 

In general, an automorphism $U\in \G$ does not leave the subalgebra $\mathcal{W}_D\subset \mathcal{W}$ invariant. Nevertheless, it maps $\mathcal{W}_D$ to another subalgebra $\mathcal{W}_{D^U}$ associated with a new abelian connection  $D^U$ defined by the rule
\begin{equation}
    D^U a=U\circ D(U^{-1}\circ a\circ U)\circ U^{-1}=Da-[DU\circ U^{-1}, a]\,.
\end{equation}
Obviously, the abelian connections $D$  and $D^U$ have the same Weyl curvature. It is worth noting that the new connection $D^U$ may not satisfy the Weyl normalizing condition, even though the initial connection $D$ does.  

A natural question to ask is when two subalgebras $\mathcal{W}_D$, $\mathcal{W}_{D'}\subset \mathcal{W}$ 
associated with abelian connections $D$ and $D'$  are conjugate to each other by an automorphism $U\in \G$.   The following theorem answers this question for the case of Fedosov connections (\ref{DF}).

\begin{theorem}\label{DD'}
    Let $D$ and $D'$ be two Fedosov connections with curvatures $\Omega$ and $\Omega'$.  
    Then the subalgebras $\mathcal{W}_D$ and  $\mathcal{W}_{D'}$ are conjugate to each other by an automorphism $U\in \G$ iff 
$[\Omega]=[\Omega']$. 
\end{theorem}

\begin{proof}  The necessity of the equality $[\Omega]=[\Omega']$ has been proven above. Let us turn to sufficiency. We can assume, by making a shift $r\rightarrow r +\theta$ with $\theta\in \Lambda^1(M)[\![\nu]\!]$ if necessary, that
that $\Omega=\Omega'$. Write  
$$
D'a=Da+\frac1\nu[\rho, a]
$$
for some 
$
\rho=\rho_n+\rho_{n+1}+\cdots
$,
where $\rho_k\in \mathcal{W}_{k,1}$, $\rho|_{y=0}=0$, and $n\geq 2$.
Then the Weyl curvature of $D'$ takes the form
$$
\Omega'=\Omega+D\rho+\frac1\nu \rho\circ\rho\,.
$$
  To leading order this gives $\delta\rho_n=0$. Since $\rho$ is assumed to satisfy the Weyl normalizing condition, Proposition \ref{P31} ensures that $\rho_n=\delta H$ for some $H\in \mathcal{W}_{n+1,0}$. Define now the automorphism 
$
U_1=\exp_\circ(\nu^{-1}H)\in \G
$ and set $D_1=D^{U_1}$. By construction, 
$$
D' a=D^{U_1}a+\frac1\nu[\rho',a]
$$
for some $\rho'\in F^{n+1}\mathcal{W}$. Without loss in generality we may again assume that $\rho'|_{y=0}=0$. Replacing now $D$ with $D_1$ and repeating the above reasoning  for the pair of Fedosov connections $D'$ and $D_1$, we construct a new connection $D_2=D_1^{U_2}$
which differs from $D'$ by an element $\rho''$ of $F^{n+2}\mathcal{W}$. Define by induction $D_{n+1}=D_n^{U_{n+1}}$ and  set
$$
    U=\prod_{k=1}^\infty U_k\,.
    $$
    Then $D_\infty=D^U=D'$, and the algebras $\mathcal{W}_D$ and $\mathcal{W}_{D'}$ are isomorphic to each other. 
\end{proof}

\section{Deformation quantization of contact manifolds}

\begin{definition}\label{D41}
A {\it deformation quantization} of a contact manifold $(M, \lambda)$ is a pair $(\ast, \Delta )$, where $\ast$ and  $\Delta$ are, respectively, bilinear and linear operators on   $C^\infty(M)[\![\nu]\!]$ of the form
$$ a\ast b= ab+\frac\nu2\pi(da\wedge db)+\sum_{k=2}^\infty \nu^kC_k(a,b)\,,\qquad
    \Delta a=\xi a+\sum_{k=1}^\infty \nu^k\Delta_k a\,.$$
It is also supposed that 
\begin{enumerate}
    \item $C_k$ and $\Delta_k$ are  differential operators on $M$, which are extended to $C^\infty(M)[\![\nu]\!]$ by $\mathbb{R}[\![\nu]\!]$-linearity;
    \item  the elements of the subspace $\ker \Delta$ form an associative algebra w.r.t. the $\ast$-product, that is, 
$$
\Delta(a\ast b)=0\quad \mbox{and}\quad  (a\ast b)\ast c=a\ast(b\ast c)
$$
whenever $\Delta a=\Delta b=\Delta c=0$. 
\end{enumerate}
The pair $(\ker \Delta, \ast)$ is called the {\it algebra of quantum observables}. 
\end{definition}

\begin{definition}
Two deformation quantizations $(\ast, \Delta)$ and $(\ast', \Delta')$ of the same contact manifold $(M, \lambda)$ are said to be {\it equivalent}, if there 
exists a formally invertible operator 
\begin{equation}\label{G}
G=1+\sum_{k=1}^\infty \nu^kG_k\,,
\end{equation}
$G_k$ being linear differential operators on $M$,
such that 
$$
\Delta' =G^{-1}\Delta G\quad\mbox{and}\quad a\ast' b=G^{-1}(G a\ast Gb)\,.
$$
\end{definition}

The following lemma will be very useful in the sequel. 
\begin{lemma}\label{L43}
For any $a\in \ww$ the equation
\begin{equation}\label{ab}
    b=a+\delta^{-1}(D+\delta)b
\end{equation}
has a unique solution for $b\in \ww$. Furthermore, $Pb=Pa$.   
\end{lemma}

The proof is very simple. As the operator  $\delta^{-1}(D+\delta)$ raises the $y$-degree by 1, one can solve this equation by iterations to get a unique solution $b$. Taking now the projection of both the sides of (\ref{ab}) yields $Pb=Pa$, since the result of applying $\delta^{-1}$ contains only positive powers of $y$'s.    

The assignment $a\mapsto b$, provided by the lemma above,  defines an endomorphism of the space $\ww$. Restricting this endomorphism onto the subspace $C^{\infty}(M)[\![\nu]\!]$ gives an $\mathbb{R}[\![\nu]\!]$-linear map 
\begin{equation}\label{Q}
Q: C^{\infty}(M)[\![\nu]\!]\rightarrow \ww\,,\end{equation}
which we call the {\it quantum lift}. This map will be essential for all our subsequent considerations. It follows from Lemma (\ref{L43}) that the projection $P$ is left inverse to $Q$, i.e., $P Q=\mathrm{id}$. Hence, the map (\ref{Q}) is injective and defines an embedding of the space $C^\infty(M)[\![\nu]\!]$ into $\ww$. Let $\tilde\nabla_i=P_i^j\nabla_j$ denote the `$\xi$-transverse covariant derivative'. Then, iterating Eq. (\ref{ab}) for $a\in C^\infty(M)[\![\nu]\!]$, one can find 
\begin{equation}\label{Qa}
\begin{array}{rcl}
Qa&=&\displaystyle a+\tilde\nabla_iay^i+\frac12\tilde\nabla_i\tilde\nabla_ja y^iy^j+\frac16 \tilde\nabla_i\tilde\nabla_j\tilde\nabla_kay^iy^jy^k\\[3mm]&-&\displaystyle \frac1{24}R_{ijkl}P^k_n\pi^{lm}\nabla_may^iy^jy^n+
\nu(\Omega_1)_{ij}P^i_kP^j_ly^k\pi^{lm}\nabla_ma+\ldots\,,
\end{array}
\end{equation}
where the dots stand for the terms of higher $y$-degree. Notice that $Q1=1$. 

Define an operator $\Delta: C^\infty(M)[\![\nu]\!]\rightarrow C^\infty(M)[\![\nu]\!]$ as the composition of maps
 \begin{equation}\label{PDQ}
     \Delta=i_\xi PD Q\,.
     \end{equation}
Now we are ready to describe the structure of the space $\ww_D$.

\begin{theorem}
    Each element $a\in \ww_D$ is completely determined by its projection $Pa\in C^\infty(M)[\![\nu]\!]$. 
    An element $a_0\in C^\infty(M)[\![\nu]\!]$ is the projection of a flat section $a\in \ww_D$ if and only if $\Delta a_0=0$.  
\end{theorem}
\begin{proof}
If $a$ is a flat section, then $\delta a=(D+\delta)a$. Applying the operator  $\delta^{-1}$ to both the sides of the last equation and using (\ref{HD}), one can find
\begin{equation}\label{Pa}
a=Pa+\delta^{-1}(D+\delta)a\,.
\end{equation}
By Lemma \ref{L43}, $a=QPa$. Hence, each flat section  $a$ is unambiguously restored by its projection $Pa$ with the help of the quantum lift (\ref{Q}).   This proves the first part of the theorem. Since each flat section $a$ obeys the equation $Da=0$, we can write 
$$
0=i_\xi P Da=i_\xi PDQPa=\Delta Pa\,. 
$$
Therefore the projection  $Pa$ of every flat section $a$ belongs to $\ker \Delta$. Conversely, given an element $a_0\in \ker \Delta$, define $a=Qa_0$. 
We claim that $a$ is a flat section. Indeed, using (\ref{Pa}) and (\ref{HD}), we find 
$$
\delta^{-1}Da=\delta^{-1}(D+\delta)a-\delta^{-1}\delta a= a-Pa-\delta^{-1}\delta a=\delta\delta^{-1}a=0\,.
$$
Next, $D^2a=0$ because $D$ is an abelian connection. Writing the last identity as $\delta Da=(D+\delta)Da$ and applying to it the operator $\delta^{-1}$, we get 
$$
\begin{array}{rcl}
Da&=&PDa+\delta^{-1}(D+\delta)Da=P^2 D Qa_0+\delta^{-1}(D+\delta)Da\\[3mm]
&=&\lambda i_\xi P D Qa_0+\delta^{-1}(D+\delta)Da=\lambda \Delta a_0+\delta^{-1}(D+\delta)Da=\delta^{-1}(D+\delta)Da\,.
\end{array}
$$
Here we used Eqs. (\ref{Pa}), (\ref{P}), and the assumption that $\Delta a_0=0$.  Hence, $Da$ has to satisfy the linear homogeneous equation $Da=\delta^{-1}(D+\delta)Da$. 
By Lemma \ref{L43}, $Da=0$.  
\end{proof}

\begin{cor}\label{C45}
As linear spaces, $\ww_D\simeq \ker \Delta$. 
\end{cor}

The quantization map (\ref{Q}) allows one to transfer the $\circ$-product on $\ww$ to the space $C^\infty(M)[\![\nu]\!]$. For all $a,b\in C^\infty(M)[\![\nu]\!]$ we set
\begin{equation}\label{star}
    a\ast b=P(Qa\circ Qb)\,.
\end{equation}
Notice that the $\ast$-product algebra $C^\infty(M)[\![\nu]\!]$ is not associative but 
contains an associative subalgebra $\ker \Delta$.  By Corollary \ref{C45}, the latter is isomorphic to the algebra of flat sections $\ww_D$. 

Now it is easy to check that the operators (\ref{PDQ}) and (\ref{star})  satisfy all the conditions of Definition \ref{D41},  thereby defining a deformation quantization of the contact manifold $(M, \lambda)$. 
We say that the deformation quantization defined by the quantum lift (\ref{Q}) is {\it canonical} if the Weyl curvature $\Omega$ of the underlying Fedosov connection $D$ is equal to $\omega$.

On substituting expressions  (\ref{rR}) and (\ref{Qa})  into (\ref{PDQ}), one can find that for the canonical deformation quantization
\begin{equation}\label{R}
    \Delta=\xi+\frac{\nu^2}{24}R_\xi^{ijk}\big(\tilde\nabla_i \tilde\nabla_j \tilde\nabla_k+\frac12R_{ijkl}\pi^{ln}\nabla_n \big)+O(\nu^3)\,,
\end{equation}
where 
$$
R_\xi^{ijk}=\frac13(R_{nmls}+R_{mlns}+R_{lnms})\xi^s\pi^{ni}\pi^{mj}\pi^{lk}\,.
$$
Equations $Q1=1$ and $D1=0$ imply $1\in \ker \Delta$. Therefore the $\ast$-product algebra $\ker \Delta$
is unital.

\begin{remark}
The Jacobi bracket (\ref{bracket}) on a contact manifold is a special case of {\it weak Poisson structures} introduced in  Ref. \cite{Lyakhovich_2005}. In that paper, a systematic method was developed for the deformation quantization of weak Poisson structures by means of Kontsevich's formality theorem. 
The method results in an $A_\infty$-algebra $(A, m_0, m_1, m_2, \ldots)$, which may be generally non-flat. For the Jacobi bracket (\ref{bracket}) the zeroth structure map $m_0$, representing an obstruction to quantization,  should vanish from dimensional considerations. Then the first structure map $m_1: A\rightarrow A$ squares to zero and makes $A$ into a cochain complex. (In the physical terminology, $m_1$ is a quantum BRST operator, while $m_0$ represents a quantum anomaly.) The cohomology of the differential $m_1$ in degree zero, $H^0(A)$, is identified with the space of quantum observables. The second structure map $m_2$ endows $A$ with a bilinear product, which is differentiated by $m_1$. On passing to cohomology,  the (non-associative) product $m_2$ induces a desired associative multiplication in the space of physical observables $H^0(A)$.  
One may see that the equation $\Delta a=0$ is essentially equivalent to and results from the cocycle condition $m_1(a)=0$ in the algebra $A$. In principle, this allows one to reconstruct the differential operator (\ref{R}) in terms of Kontsevich's integrals.   However, Fedosov's approach  offers a more direct and simpler way to solve this problem. 
\end{remark}

Let us examine more closely the space of quantum observables $\ker \Delta$. On substituting the power series expansions for the operator $\Delta$ and an element $a\in C^{\infty}(M)[\![\nu]\!]$, 
$$\Delta=\xi+\nu^2\Delta_2+\cdots \,,\qquad a=a_0+\nu a_1+\nu^2 a_2+\cdots\,,$$ into  $\Delta a=0$, we obtain an infinite sequence of equations 
\begin{equation}\label{qcor}
\xi a_0=0\,,\qquad \xi a_1=0\,,\qquad \xi a_2+\Delta_2 a_0=0\,,\qquad\ldots
\end{equation}
The first equation says that the leading term $a_0$ is given by a $\xi$-invariant function on $M$. As discussed in Sec. \ref{S2} such functions form a Poisson algebra  $C^\infty_\xi(M)$ w.r.t. the Jacobi bracket (\ref{bracket}). The Poisson algebra $(C^\infty_\xi(M), \{\,\cdot\,,\,\cdot\,\})$ is naturally identified with the algebra of {\it classical observables}. One may also regard the higher powers of $\nu$ in  $a$ as `quantum corrections' to a classical observable $a_0$. The corrections are not arbitrary but satisfy the infinite 
system of linear equations (\ref{qcor}). The role of these equations is twofold: they provide the closeness 
of the subspace $\ker \Delta\subset C^\infty(M)[\![\nu]\!]$ under the $\ast$-product (\ref{star}) {\it and}  the associativity of the $\ast$-product algebra $\ker \Delta$.
In contrast to the deformation quantization of symplectic manifolds, not every classical observable $a_0$ on a contact manifold admits a consistent quantization, because the solvability of system (\ref{qcor}) is not ensured in advance. 
A natural question arises: whether the property of a classical observable to be or not to be quantizable depends on the choice of a contact connection that enters the operator $\Delta$? The answer is no, since all deformation quantizations turn out to be equivalent to each other whenever they correspond to the same Weyl curvature. By Theorem \ref{DD'}, two Fedosov connections $D$ and $D'$ associated with different contact connections $\nabla$ and $\nabla'$ but the same Weyl curvature $\Omega$ give rise to the isomorphic algebras of flat sections $\mathcal{W}_D$ and $\mathcal{W}_{D'}$ that are related by an automorphism $U\in \G$ of  $\ww$. The algebras $\mathcal{W}_D$ and $\mathcal{W}_{D'}$, in their turn,  are isomorphic to the algebras $\ker \Delta$ and $\ker \Delta'$, respectively (Corollary \ref{C45}).  This allows us to define the desired equivalence  (\ref{G}) between the deformation quantizations $(\ast, \Delta)$ and $(\ast', \Delta' )$ as the composition of maps
\begin{equation}\label{Ga}
Ga=P(U\circ Qa\circ U^{-1})\qquad \forall a\in C^\infty(M)[\![\nu]\!]\,.
\end{equation}
On restricting to $\ker \Delta$ it gives the commutative diagram 
$$
\xymatrix{\mathcal{W}_D\ar[d]_U &\ker\Delta \ar[l]_Q\ar[d]^G\\
\mathcal{W}_{D'}\ar[r]^P&\ker\Delta'}
$$
where the right arrow defines an isomorphism of the algebras of quantum observables. 

As discussed in Sec.~2, the difference of two contact connections on $\mathcal{T}^\xi(M)$ is given by a totally symmetric $\xi$-transverse tensor $S_{ijk}$, see Eq. (\ref{Sdef}). 

With the recurrent relations above it is not hard to find that  for the canonical deformation quantization ($\Omega=\omega$) the automorphism $U$ and the equivalence operator (\ref{Ga}) take the form
$$ 
 U=\exp_\circ (\nu^{-1}H)\,,  \qquad G =1+ \nu^2G_2+O(\nu^3)\,,
$$
where
$$
H=\frac1{6}S_{ijk}y^iy^jy^k- \frac1{24} \tilde \nabla_n S_{ijk}y^iy^jy^ky^n +  \frac{\nu^2}{32} S^{ijk} R_{ijkl} P^l_{n} y^{n} + \dots\,,\quad S^{ijk}=S_{lnm}\pi^{li}\pi^{nj}\pi^{mk}\,,
$$
and
\begin{equation}
   G_2=\frac{1}{24}S^{ijk} \Big(\tilde \nabla_i\tilde \nabla_j \tilde \nabla_k + \frac32 S^l_{ij} \tilde \nabla_k \tilde\nabla_l + \frac12 R_{ijkl} \pi^{ln} \nabla_{n} - \frac14  \nabla_l S_{ijk} \pi^{ln} \nabla_{n}   \Big) \,.
\end{equation}
As a consequence 
\begin{equation}\label{d'}
\Delta'=G^{-1}\Delta G=\Delta+\nu^2[\xi, G_2]+O(\nu^3)\,.
\end{equation}
Straightforward calculations yield
\begin{equation}
        \Delta'_2 =  \frac{1}{24} \big( R^{ijk}_\xi + \nabla_\xi S^{ijk} \big) \Big( \tilde\nabla_i \tilde\nabla_j \tilde\nabla_k + \frac12 R_{ijk}{}_m \pi^{mn} \nabla_{n} - \nabla_l S_{ijk} \pi^{ln} \nabla_{n} \Big) \,.
\end{equation}
One could also arrived at this expression by a direct change of connection $\nabla\rightarrow \nabla+S$ in (\ref{R}).

\section{The first  obstruction to quantization}

Each classical observable $a_0\in C^\infty_\xi(M)$ gives rise to a chain of equations (\ref{qcor}) for quantum corrections $a_1$, $a_2$, ..., needed to promote $a_0$ to the quantum level.  A classical observable $a_0$ is said to be {\it quantizable} if system (\ref{qcor}) admits a solution. Otherwise, one speaks of obstructions to quantization. Evidently, the first two equations in (\ref{qcor}) define classical observables $a_0$ and $a_1$. Given $a_0$, one can always satisfy the second equation by setting $a_1=0$ (no first-order corrections). The first source of obstructions comes with the third equation  for the second-order corrections. It says that the function $-\Delta_2 a_0$ should be the gradient of some function $a_2$ along the Reeb vector field $\xi$. 
This motivates us to introduce the subspace $\mathrm{Im}\,\xi\subset C^\infty(M)$ and the corresponding quotient space $\mathcal{O}=C^\infty(M)/\mathrm{Im}\, \xi$. It is convenient to combine the mentioned spaces and maps into the exact sequence 
$$
\xymatrix{ 0\ar[r]&C^\infty_\xi(M)\ar[r]&C^\infty(M)\ar[r]^{\xi}&C^\infty(M)\ar[r]&\mathcal{O}\ar[r]&0}\,.
$$
The space $\mathcal{O}$ is generally nonzero and accommodates  all potential obstructions to the solvability  of the equation
\begin{equation}\label{1o}
-\Delta_2 a_0=\xi a_2\,.
\end{equation}

Recall that integral curves of the Reeb vector field $\xi$ are called {\it characteristics} of a contact manifold $(M,\lambda)$. The most interesting are closed characteristics, that is, periodic orbits of the Reeb flow. The famous Weinstein's  conjecture states that  every compact manifold $M$ enjoys at least one closed characteristic. By now the Weinstein conjecture is a theorem for a wide class of contact manifolds, see e.g. \cite[Ch. 3.4]{Blair}.
To each closed characteristic $\gamma$ there corresponds its period  
\begin{equation}\label{t-per}
\tau_\gamma=\int_\gamma \lambda\in \mathbb{R}\,.
\end{equation}
The periods of closed characteristics are thus invariants of contact manifolds. 

Given a  closed characteristic $\gamma\subset M$, consider the linear functional 
$$
\Psi_\gamma[a]=\int_\gamma\lambda \Delta_2 a\,.
$$
To write this functional one needs to choose a contact connection, which enters the operator $\Delta_2$. However, the dependence of the connection drops out if we restrict the functional to the subspace of classical observables $C^\infty_\xi(M)$. Indeed, using Eq. (\ref{d'}) together with the simple identities 
$$
\xi a=\lambda i_\xi da=da -i_\xi (\lambda\wedge da)\,,\qquad i_\xi\alpha|_\gamma=0\,,
$$
which hold for all $a\in C^\infty(M)$ and $\alpha\in \Lambda^2(M)$,  
we readily find that
\begin{equation}\label{psi'}
\Psi_\gamma'[a_0]=\int_\gamma\lambda\Delta'_2a_0= \Psi_\gamma [a_0]+\int_\gamma \lambda [\xi, G_2]a_0=\Psi_\gamma[a_0]+\int_\gamma\lambda\xi (G_2a_0)=\Psi_\gamma[a_0]
\end{equation}
whenever $a_0\in C^\infty_\xi(M)$. We will refer to the number $\Psi_\gamma[a_0]$ as the {\it characteristic of a classical observable} $a_0$ (associated with a closed characteristic $\gamma\subset M$). Integrating   Eq. (\ref{1o}) over  closed characteristics, we conclude that all characteristics of  quantizable classical observables must vanish. Thus, nonzero numbers $\Psi_\gamma[a_0]$ represent the first 
obstructions to the quantization of $a_0\in C^\infty_\xi(M)$.  

For the rest of this section let us suppose that the contact manifold $M$ is compact. Then we can equip the space of smooth functions $C^\infty(M)$ with a positive-definite inner product: 
$$
(a,b)=\int_Mv ab \,\in \mathbb{R}\,,
$$
$v$ being the canonical volume form on $M$. Since $L_\xi v=0$,  Stokes' theorem gives immediately  $(a,\xi b)=-(\xi a,b)$;  and hence,   $C^\infty_\xi(M)\perp \mathrm{Im}\,\xi$. Taking the inner product of  Eq. (\ref{1o}) 
with an arbitrary $b\in C^{\infty}_\xi(M)$ leads to the equality $(b,\Delta_2 a_0)=0$. For $b=1$ we obtain
\begin{equation}\label{ch-a}
\begin{array}{c}
(1, \Delta_2 a_0)=-\frac{1}{24}(\chi, a_0)=0\,,
\end{array}
\end{equation}
where $\chi=-24\Delta^\ast_2 1$  and $\Delta^\ast_2$ is the formal adjoint of the operator $\Delta_2$. Explicitly,
\begin{equation}\label{char}
\chi=\tilde \nabla_i\tilde \nabla_j\tilde \nabla_kR_\xi^{ijk}+\frac12\nabla_n (R_\xi^{ijk}R_{ijkl}\pi^{ln})\,.
\end{equation}
Eq. (\ref{ch-a}) says that each quantizable classical observable must be orthogonal to the function (\ref{char}). The transformation property of $\chi$ under the change of connection, $\nabla\rightarrow \nabla+S$, follows immediately from that of the operator $\Delta_2$. Using (\ref{d'}), one can find
\begin{equation}\label{chi'}
    \begin{array}{c}
        \chi'=\chi+\xi \psi\,,\\[3mm]
        \displaystyle \psi=\tilde\nabla_i\tilde\nabla_j\tilde\nabla_k S^{ijk}-
        \frac{3}{2}\tilde\nabla_k\tilde\nabla_l(S^{ijk}S_{ij}^l)+\frac12\nabla_n(S^{ijk}R_{ijkl}\pi^{ln})-\frac14\nabla_n(S^{ijk}\pi^{nl}\nabla_lS_{ijk})\,.
    \end{array}
\end{equation}
Formula (\ref{char}) defines then a unique element $[\chi]\in \mathcal{O}$ -- the equivalence class of $\chi$ modulo $\xi$-gradients -- which we will call the {\it character} of a contact manifold\footnote{Notice that one may omit the tildes over covariant derivatives in (\ref{R}) whenever the operator $\Delta_2$  applies  to $a\in C^\infty_\xi(M)$. In view of this, omitting tildes in  (\ref{char}) gives an equivalent element of the class $[\chi]$.}. The character $[\chi]$, being independent of the choice of a contact connection, provides an  invariant of a contact manifold $(M,\lambda)$ itself.   This also means the invariance of the orthogonality condition (\ref{ch-a}) under the choice of a representative of $[\chi]$. Notice that $[1]\neq 0$ whenever $M$ is compact. In spite of a bit abstract definition, the character gives rise to a set of more handy invariants, which are similar to the periods (\ref{t-per}). For each closed characteristic $\gamma\subset M$ we put
\begin{equation}\label{ch-per}
\chi_\gamma=\int_\gamma \lambda  \chi \,.
\end{equation}
It is clear that the value of the integral does not depend on a particular representative $\chi\in [\chi]$; and hence, on a contact connection.  We call the numbers (\ref{ch-per}) the {\it $\chi$-periods} of a contact manifold. As the $\tau$-periods (\ref{t-per}) they make sense for non-compact manifolds as well\footnote{Each closed characteristic $\gamma\subset M$ gives rise to a linear functional $\Phi: \mathcal{O}\rightarrow\mathbb{R}$ defined by the integral  $\Phi[a]=\int_\gamma \lambda a$.  The $\tau$- and $\chi$-periods are just special values of this functional that correspond to the pair of distinguished  elements $[1], [\chi]\in \mathcal{O}$.}. 
It would be interesting to further explore the geometric meaning of the character and $\chi$-periods as well as their relation, if any, to other known invariants of contact manifolds. 

Remember,  the existence of a solution to Eq. (\ref{1o}) does not yet imply that a classical observable $a_0$ is quanizable. For this to happen, all higher equations in (\ref{qcor}) must  hold as well.  This leads to higher obstructions, which may also depend on choices made for the lower order quantum corrections.  Very quickly, the analysis becomes extremely cumbersome. In this regard, it is pertinent to note a simple geometric condition that ensures the absence of any obstruction. This is the case of a $\xi$-invariant contact connection, i.e., $L_\xi\nabla=0$. Although invariant contact connections are quite rare,  Lichnerowicz has shown \cite{LICHNEROWICZ1982} how to construct such a connection starting form any affine connection on $M$ invariant under the Reeb flow $\xi$. 
Since the covariant curvature tensor of an invariant contact connection is totally $\xi$-transverse together with all its covariant derivatives,
$
\nabla_{i_1}\cdots \nabla_{i_n}R_{ijkl}\in \mathcal{T}^\xi_{n+4}(M)\,,
$
the one-form $r$ defining the Fedosov connection (\ref{DF}) may be chosen to satisfy the $\xi$-transversality condition $i_\xi r=0$ as well. As a result, all the quantum corrections in (\ref{R}) are absent and the whole operator reduces to $\Delta=\xi$. Hence, any classical observable is quantizable in this case. This conclusion is in line with the previous results on quantization of presymplectic manifolds \cite{GorElfSha:2020,Vaisman_2002,block1992quantization}. For an invariant contact connection  the operator $\Delta$ is just a derivation of the  $\ast$-product,
$$
\xi(a\ast b)=(\xi a)\ast b+a\ast(\xi b)\,,
$$
and the closeness of the subspace $\ker \Delta=C^\infty_\xi(M)[\![\nu]\!]$ under the $\ast$-product becomes obvious.

\end{document}